\documentclass[12pt,reqno]{amsart}
\usepackage{amsthm,amsfonts,amssymb,euscript}

\theoremstyle{plain}
  \newtheorem{theorem}[subsection]{Theorem}
  
  \newtheorem{proposition}[subsection]{Proposition}
  \newtheorem{lemma}[subsection]{Lemma}

\def\bphi{{\boldsymbol{\phi} }}
\def\bpsi{{\boldsymbol{\psi} }}
\def\F{{\mathcal F}}
\def\R{{\mathbb{R}}}
\def\C{{\mathbb{C}}}
\def\I{{\mathcal I}}

\def\A{{\mathcal A}}
\def\B{{\mathcal B}}
\def\P{{\mathcal P}}
\def\Q{{\mathcal Q}}
\def\H{{\mathcal H}}
\def\N{{\mathcal N}}

\def\M{{\mathcal M}}
\def\E{{\mathcal E}}
\def\V{{\mathcal V}}

\def\ch{\mbox{ch} (k)}
\def\chbb{\mbox{chb} (k)}
\def\trigh{\mbox{trigh} (k)}

\def\sh{\mbox{sh} (k)}
\def\shh{\mbox{sh}}
\def\chh{\mbox{ch}}
\def\shhb{\overline {\mbox{sh}}}
\def\th{\mbox{th} (k)}
\def\chhb{\overline {\mbox{  ch}}}
\def\shbb{\mbox{shb} (k)}
\def\chb{\overline{\mbox{\rm ch}  (k)}}
\def\shb{\overline{\mbox{sh}(k)}}
\def\z{\zeta}
\def\zb{\overline{\zeta}}
\def\FF{{\mathbb{F}}}
\theoremstyle{remark}
  \newtheorem{remark}[subsection]{Remark}
  
\theoremstyle{definition}
  \newtheorem{definition}[subsection]{Definition}

\begin{document}

\def\vac{{\big\vert 0\big>}}
\def\fpsi{{\big\vert\psi\big>}}
\def\bphi{{\boldsymbol{\phi} }}
\def\bpsi{{\boldsymbol{\psi} }}
\def\F{{\mathcal F}}
\def\R{{\mathbb{R}}}
\def\C{{\mathbb{C}}}
\def\I{{\mathcal I}}
\def\Q{{\mathcal Q}}
\def\ch{\mbox{\rm ch} (k)}
\def\cht{\mbox{\rm ch} (2k)}
\def\chbb{\mbox{\rm chb} (k)}
\def\trigh{\mbox{\rm trigh} (k)}
\def\p{\mbox{p} (k)}
\def\sh{\mbox{\rm sh} (k)}
\def\sht{\mbox{\rm sh} (2k)}
\def\shh{\mbox{\rm sh}}
\def\chh{\mbox{ \rm ch}}
\def\th{\mbox{\rm th} (k)}
\def\thb{\overline{\mbox{\rm th} (k)}}
\def\shhb{\overline {\mbox{\rm sh}}}
\def\chhb{\overline {\mbox{\rm ch}}}
\def\shbb{\mbox{\rm shb} (k)}
\def\chbt{\overline{\mbox{\rm ch}  (2k)}}
\def\shb{\overline{\mbox{\rm sh}(k)}}
\def\shbt{\overline{\mbox{\rm sh}(2k)}}
\def\z{\zeta}
\def\zb{\overline{\zeta}}
\def\FF{{\mathbb{F}}}
\def\S{{\bf S}}
\def\W{{\bf W}}
\def\WW{{\mathcal W}}
\def\N{{\bf N}}
\def\M{{\bf M}}
\def\L{{\mathcal L}}
\def\E{{\mathcal E}}
 \title[]%
{ Beyond mean field: on the role of pair excitations in the evolution of condensates.}

\author{M. Grillakis}
\address{University of Maryland, College Park}
\email{mng@math.umd.edu}

\author{M. Machedon}
\address{University of Maryland, College Park}
\email{mxm@math.umd.edu}

\thanks{
The authors thank
Sergiu Klainerman for the interest shown for this work, and
John Millson for many
discussions related to  the symplectic group and its representations. The authors would like to thank the Department of Applied Mathematics at the University of Crete and ACMAC for their hospitality during the
preparation of the present work.
}

\date{}
\dedicatory{This paper is dedicated to Professor Choquet-Bruhat, whose work
on nonlinear wave equations was an inspiration to us earlier in our careers}
\commby{}
\maketitle
\section{Introduction}
\begin{abstract}
This paper is in part a summary of our earlier work \cite{GMM1, GMM2, GM}, and in part an announcement  introducing a refinement
of the equations for the pair excitation function used in our previous work with D. Margetis. The new equations are Euler-Lagrange equations, and the solutions conserve energy and the number of particles.
\end{abstract}
\section{Introduction}

The problem, which has received a lot of attention in recent years, is concerned with the evolution of the $N$-body linear Schr\"odinger equation
\begin{align*}
&\frac{1}{i} \frac{\partial}{\partial t} \psi_N(t, \cdot)= H_N\psi_N(t, \cdot) \mbox{ with}\\
&\psi_N (0, x_1, \cdots, x_N) =\phi_0( x_1) \phi_0( x_2) \cdots
\phi_0( x_N)\\
&  \|\psi_N(t, \cdot)\|_{ L^2(\mathbb R^{3N})}=1
\end{align*}
The Hamiltonian is an operator of the form
 \begin{align*}
H_N=\sum_{j=1}^N \Delta_{x_j}-\frac{1}{N}\sum_{i<j}v_N(x_i-x_j)
\end{align*}
where $v_{N}(x):=N^{3\beta}v(N^{\beta}x)$ with $0\leq\beta\leq 1$
models the strength of two body interactions. Notice that if $\beta >0$ then $v_{N}(x)\to\delta(x)$ as $N\to\infty$. For simplicity we assume that
$v\in C_{0}$ and $v\geq 0$.
The goal is to show, in a sense to be made precise,
\begin{align}
\psi_N (t, x_1, \cdots, x_N) \simeq e^{i N \chi(t)}\phi(t, x_1) \phi(t, x_2) \label{manyN} \cdots
\phi(t, x_N)
\end{align}
where
$\phi$ satisfies a suitable non-linear Schr\"odinger equation. In particular, this approximation is not true in
$L^2(\mathbb R^{3N})$.

The motivation for this problem is that in the presence of a trap the ground state of $H_N$
 looks like
 \begin{align*}
\Psi_N (x_1, x_2, \cdots, x_N) \backsimeq \phi_0( x_1) \phi_0( x_2) \cdots
\phi_0( x_N)
\end{align*}
This is suggested by the result of Lieb and Seiringer who showed in \cite{LS1} that
\begin{align*}
\gamma_1^N(x, x') \to \phi_0(x)\overline{\phi_0}(x')
\end{align*}
where
\begin{align*}
\gamma_1^N(x, x') = &\int \Psi_N (x, x_2, \cdots, x_N)\overline{\Psi_N} (x', x_2, \cdots, x_N)
dx_2 \cdots dx_N
\end{align*}
Here $\|\phi_0\|_{L^2}=1$ and $\phi_0$ minimizes the Gross-Pitaevskii functional. See \cite{lieb05} for  extensive
background.

The reason for the recent attention to this problem is two-fold. On the one hand experimental advances during the last twenty years made the
creation and manipulation of condensates in the laboratory possible, on the other hand recent mathematical developments made possible the
rigorous treatment of the equations when the number of particles, namely $N$, is large.

While this is a "classical PDE problem" (as opposed to a Fock space problem), the PDE approach to this problem
 has only been studied systematically  during the last 10-15 years, in the series of papers of
 Erd\"os and Yau \cite{E-Y1}, and
 Erd\"os, Schlein and Yau \cite{E-S-Y1} to \cite{E-S-Y4}. See also \cite{E-E-S-Y1}. These papers prove
 \begin{align}
\gamma_1^N(t, x, x') \to \phi(t, x)\overline{\phi}(t, x') \label{ESY}
\end{align}
in trace norm as $N \to \infty$, and similarly for the higher order marginal density matrices $\gamma_k^N$, where $k$ is fixed. The problem becomes more difficult and interesting as the parameter $\beta$ in the definition of $v_N$ approaches 1.
The strategy of these papers is based on the older work  of Spohn \cite{spohn}.
 Recent simplifications and generalizations,
 based on harmonic analysis techniques and a "boardgame argument" inspired by the Feynman diagram approach of
 Erd\"os, Schlein and Yau,
 were given in \cite{K-MMM}, \cite{KSS}, \cite{CP}, \cite{XC3},     \cite{CH}, \cite{CH1}. See also
 \cite{FKS}, \cite{K-P} for a different approach.

The symmetric Fock space approach to the problem is much older.
It originated in physics, with the papers by Lee, Huang and Yang \cite{Lee-H-Y} in the static case, and Wu \cite{wuI}
in the time-dependent case.
 See also \cite{Bo}.
 It continued with the mathematically rigorous work of Hepp \cite{hepp}, and Ginibre and Velo \cite{G-V}.

Motivated by the goal of obtaining a convergence rate to solutions of NLS in \eqref{ESY}, Rodnianski and Schlein resumed the rigorous Fock space approach in \cite{Rod-S}. This paper, as well as the older work of Wu, served as an inspiration
for our work. Our goal is to obtain a refinement to \eqref{manyN} which provides an $L^2(\mathbb R^{3N})$ and Fock space estimate.
This leads to the introduction of the pair excitation function $k$.

We also mention the recent preprint \cite{B-O-S} where a similar approach (but with an explicit choice of pair excitation function $k$) is used to prove convergence of the density matrices in the critical case $\beta=1$.

\section{Fock space}

In this section we briefly review symmetric Fock space, following the notation of \cite{GM}.  See \cite{Rod-S}, for more details. The elements of
$\FF$ are  vectors of the form
\begin{equation*}
\big\vert\psi\big>=\big(\psi_{0}\ ,\ \psi_{1}(x_{1})\ ,\ \psi_{2}(x_{1},x_{2})\ ,\ \ldots\ \big)
\end{equation*}
where $\psi_{0}\in \C$ and $\psi_k$ are symmetric $L^2$ functions. The norm of such a vector is,
\begin{equation*}
\big\Vert\ \big\vert\psi\big>\big\Vert_{\F}^{2}=\big<\psi\big\vert\psi\big> =\vert\psi_{0}\vert^{2}+
\sum_{n=1}^{\infty}\big\Vert\psi_{n}\big\Vert^{2}_{L^{2}}\ .
\end{equation*}
The creation and anihilation distribution valued operators denoted by $a^{\ast}_{x}$ and
$a_{x}$ respectively which act on vectors of the form $(0, \cdots, \psi_{n-1}, 0,  \cdots)$ and
$(0, \cdots, \psi_{n+1}, 0,  \cdots)$
by
\begin{align*}
&a^{\ast}_{x}(\psi_{n-1}):=\frac{1}{\sqrt{n}}\sum_{j=1}^{n}\delta(x-x_{j})
\psi_{n-1}(x_{1},\ldots ,x_{j-1},x_{j+1},\ldots ,x_{n})
\\
&a_{x}(\psi_{n+1}):=\sqrt{n+1}\psi_{n+1}([x],x_{1},\ldots ,x_{n})
\end{align*}
with $[x]$ indicating that the variable $x$ is frozen.
The vacuum state is defined as follows:
\begin{equation*}
\vac :=(1,0,0\ldots )
\end{equation*}
and $a_{x}\vac =0$. One can easily check that $\big[a_{x},a^{\ast}_{y}\big]=\delta(x-y)$ and since
the creation and anihilation operators are distribution valued we can form operators that act on $\FF$ by introducing
a field, say $\phi(x)$, and form
\begin{align*}
a (\bar{\phi}):=\int dx\left\{\bar{\phi}(x)a_{x}\right\}
\quad {\rm and}\quad
a^{\ast} (\phi):=\int dx\left\{\phi(x)a^{\ast}_{x}\right\}
\end{align*}
where by convention we associate $a$ with $\bar{\phi}$ and  $a^{\ast}$ with $\phi$.
These operators are well defined, unbounded, on $\FF$ provided that $\phi$ is square integrable.
The creation and anihilation operators provide a way to introduce coherent states in $\FF$ in the following manner,
first define the skew-Hermitian operator
\begin{equation}
\A(\phi):=\int dx\left\{\bar{\phi}(x)a_{x}-\phi(x)a^{\ast}_{x}\right\} \label{meanfield-2}
\end{equation}
and then introduce $N$-particle coherent states as
\begin{equation}\label{cohstates}
\big\vert\psi(\phi)\big>:=e^{-\sqrt{N}\A(\phi)}\vac \ .
\end{equation}
This is the Weyl operator used by Rodnianski and Schlein in \cite{Rod-S}.
It is easy to check that
\begin{equation*}
e^{-\sqrt{N}\A(\phi)}\vac =\left(\ldots\ c_{n}\prod_{j=1}^{n}\phi(x_{j})\ \ldots\right)
\quad {\rm with}\quad  c_{n}=\big(e^{-N}N^{n}/n!\big)^{1/2}\ .
\end{equation*}
In particular, by Stirling's formula, the main term that we are interested in has the coefficient
\begin{align}
c_{N}\approx (2\pi N)^{-1/4} \label{stirling}
\end{align}
Thus a coherent state introduces a tensor product in each sector of $\FF$.

For the  construction analogous to \eqref{meanfield-2} involving quadratics, start with the Lie algebra  of real or complex
 symplectic "matrices" of the  form
\begin{equation}
L:=\left(\begin{matrix} \notag
d(x,y)&l(x,y)\\
k(x,y)&-d^T(x,y)
\end{matrix}\right)
\end{equation}
where $d$, $k$ and $l$ are kernels in $L^{2}$, and $k$ and $l$ are symmetric in $(x,y)$. We denote this Lie algebra
$sp(\mathbb C)$ or $sp(\mathbb R)$ depending on whether the kernels $d$, $k$ and $l$ are complex or real.
The natural setting for us (which will insure that the Fock space operator $e^{\I(L)}$ defined below, is unitary, see also the appendix of \cite{GMM2}) is the subalgebra  $ sp_c(\mathbb R)=\WW sp(\mathbb R) \WW^{-1}$ where
\begin{align*}
\WW=
\frac{1}{\sqrt 2}
\left(
\begin{matrix}
I & i I\\
I & -iI
\end{matrix}
\right)
\end{align*}
The elements of $ sp_c(\mathbb R)$ look like
\begin{equation}
L:=\left(\begin{matrix}
i d(x,y)& \overline{k}(x,y)\\ \label{spc}
k(x,y)&-i d^T(x,y)
\end{matrix}\right)
\end{equation}
with $L^2$ kernels $d$ complex and self-adjoint,  $k$ complex and symmetric.
\begin{remark} \label{E}
The corresponding group elements $E \in Sp_c(\mathbb R) $ (in particular $E = e^{L}, \, L \in  sp_c(\mathbb R) $)  satisfy  the following three properties:
\begin{itemize}
\item
 $E$ commutes with the real structure $\sigma$ defined by $\sigma (
\phi,
\psi)=
(\overline{\psi},
\overline{\phi})$, in other words $E$ is of the form
\begin{align*}
E=\left(\begin{matrix}
P(x,y)&Q(x,y)\\
\overline{Q}(x,y)&\overline{P}(x, y)
\end{matrix}\right)
\end{align*}
\item
$E$  belongs to the infinite dimensional analogue of $U(n, n)$ , in other words
\begin{align*}
E^*
\left(
\begin{matrix}
I & 0\\
0 & -I
\end{matrix}
\right)
E =
\left(
\begin{matrix}
I & 0\\
0 & -I
\end{matrix}
\right)
\end{align*}
\item $E$ is in the symplectic group, meaning
\begin{align*}
E^T
\left(
\begin{matrix}
0 & I\\
-I & 0
\end{matrix}
\right) E=
\left(
\begin{matrix}
0 & I\\
-I & 0
\end{matrix}
\right)
\end{align*}
\end{itemize}
In fact, any two of the above imply the third. The conceptual reason for this is that
 the symplectic inner product $\bigg((\phi_1, \psi_1), (\phi_2, \psi_2) \bigg)=\int \phi_1 \psi_2 - \int  \psi_1 \phi_2$
 and the "$U(n, n)$" inner product  $\bigg<(\phi_1, \psi_1), (\phi_2, \psi_2) \bigg>=\int \phi_1 \overline{ \phi_2} - \int
 \psi_1 \overline{ \psi_2}
 $ are related by  $\bigg<(\phi_1, \psi_1), (\phi_2, \psi_2) \bigg>=\bigg((\phi_1, \psi_1), \sigma (\phi_2, \psi_2) \bigg)$.
 See Folland's book \cite{F} for more along these lines in the finite dimensional case.These matrices are called Bogoliubov rotations in
\cite{B-O-S}.
\end{remark}

Our approach is based on the map from $L \in sp(\mathbb C)$ to
quadratic polynomials in $(a,a^{\ast})$ in the following manner,
\begin{align}\label{liemap}
&\I\big(L\big) =
\frac{1}{2}\int dxdy\left\{
(a_{x}\ ,\ a^{\ast}_{x})\left(\begin{matrix}d(x,y)&l(x,y)\\
k(x,y)&-d(y,x)\end{matrix}\right)\left(\begin{matrix}-a^{\ast}_{y}\\ a_{y}\end{matrix}\right)\right\}\\
&=-\frac{1}{2}\int dxdy\left\{
d(x,y) a_{x}a^{\ast}_{y}+d(y,x)a^*_x a_y+k(x,y)a_x^*a_y^*-l(x,y)a_xa_y\right\} \notag\ .
\end{align}
This is the infinite dimensional Segal-Shale-Weil infinitesimal representation. The group representation was studied in \cite{shale}.
The crucial property of this map is the Lie algebra isomorphism
\begin{equation}\label{Lieisomorph}
\big[\I(L_{1}),\I(L_{2})\big]=\I\big([L_{1},L_{2}]\big)
\end{equation}
Notice that if $L \in sp_c(\mathbb R)$, then $L$ has the form \eqref{spc} and $\I(L)$ is skew-Hermitian, thus
  $e^{\I(L)} $ is a unitary operator on Fock space.
For the applications that follow we will  only use the self-adjoint elements of $sp_c(\mathbb R)$
\begin{align}
K=\left(\begin{matrix}
0 & \overline{k}(t, x, y)\\ \label{K}
k(t, x, y)& 0
\end{matrix}
\right)
\end{align}
and the corresponding
\begin{align}
&\B(k):=\I (K)=
\frac{1}{2}\int dxdy\left\{\bar{k}(t,x,y)a_{x}a_{y}
-k(t,x,y)a^{\ast}_{x}a^{\ast}_{y}\right\}\ . \label{pairexcit-2}
\end{align}
\begin{align}
K=\left(\begin{matrix}
0 & \overline{k}(t, x, y)\\ \label{K}
k(t, x, y)& 0
\end{matrix}
\right)
\end{align}
We easily compute
\begin{align*}
e^K= \left(
\begin{matrix}
\ch & \shb\\
\sh & \chb
\end{matrix}
\right)
\end{align*}
where
\begin{subequations}
\begin{align}
\sh &:=k + \frac{1}{3!} k \circ \overline k \circ k + \ldots~, \label{hypebolicsine}\\
\ch &:=\delta(x-y) + \frac{1}{2!}\overline k  \circ k + \ldots~,\label{hypebolicosine}
\end{align}
\end{subequations}
This particular construction and the corresponding unitary operator $e^{\B}$ were introduced in \cite{GMM1}.

The Fock Hamiltonian  is
\begin{subequations}
\begin{align}
&\H:=\H_{1}-N^{-1}\V\quad\quad\quad   {\rm where,} \label{FockHamilt1-a}\\
&\H_{1}:=\int dxdy\left\{
\Delta_{x}\delta(x-y)a^{\ast}_{x}a_{y}\right\}\quad {\rm and} \label{FockHamilt1-b}\\
&\V:=\frac{1}{2}\int dxdy\left\{v_{N}(x-y)a^{\ast}_{x}a^{\ast}_{y}a_{x}a_{x}\right\}.
\label{FockHamilt1-c}
\end{align}
\end{subequations}
It is a diagonal operator on Fock space, and it acts as a regular PDE Hamiltonian in $n$ variable
\begin{align*}
H_{n, \, \, PDE}=\sum_{j=1}^{n}\Delta_{x_{j}} - \frac{1}{2N}
\sum_{x_{j}\not= x_{k}}N^{3\beta}v\big(N^{\beta}(x_{j}-x_{k})\big)
\end{align*}
 on the $n$th component of $\FF$.

\section{Outline of older results}

Our goal is to study the evolution of coherent initial conditions of the form
\begin{align}
 \big\vert\psi_{exact}\big>=e^{i t \H}e^{-\sqrt{N}\A(\phi_{0})}\big\vert 0\big> \label{exact}
 \end{align}
 The papers \cite{GMM1, GMM2, GM} propose an approximation of the form
\begin{equation}\label{approx}
\big\vert\psi_{appr}\big>:=e^{-\sqrt{N}\A(\phi(t))}e^{-\B(k(t))}\big\vert 0\big>,
\end{equation}
and derive Schr\"odinger type equations  for $\phi(t, x)$, $k(t, x, y)$
so that
 $\big\vert\psi_{exact}(t)\big>\approx e^{iN\chi(t)}\big\vert\psi_{appr}(t)\big>$,
 with $\chi(t)$ a real phase factor, and find
  precise estimates in Fock space, see Theorem \eqref{GMthm} below.
  Our strategy is to consider
  \begin{align*}
  \big\vert\psi_{red}\big>=e^{\B(t)}e^{\sqrt{N}\A(t)}e^{it\H}e^{-\sqrt{N}\A(0)}\big\vert 0\big>
  \end{align*}
and then  find a "reduced Hamiltonian" $H_{red}$ so that
\begin{equation}
\frac{1}{i}\partial_{t}\big\vert\psi_{red}\big>=
\H_{red}\big\vert\psi_{red}\big> \ .
\end{equation}
The reduced Hamiltonian is
\begin{align*}
\H_{red} &:=\frac{1}{i}\big(\partial_{t}e^{\B}\big)e^{-\B}
\\
&+e^{\B}\left(\frac{1}{i}\big(\partial_{t}e^{\sqrt{N}\A}\big)e^{-\sqrt{N}\A}
+e^{\sqrt{N}\A}\H e^{-\sqrt{N}\A}\right)e^{-\B}\ .
\end{align*}
It can be written abstractly as a  composition (in space only) of operators
\begin{align*}
\H_{red}=\frac{1}{i}\frac{\partial}{\partial t} + e^{\B}e^{\sqrt N \A}\left(-\frac{1}{i}\frac{\partial}{\partial t} + \H\right)\circ e^{-\sqrt N \A}e^{-\B}
\end{align*}
Explicitly it is
\begin{align}
&\H_{red}
= N \P_0 + N^{1/2}e^{\B} \P_1 e^{-\B} \notag
\\
&+\H_G +\I(R)-N^{-1/2}e^{\B} \P_3e^{-\B}-N^{-1}e^{\B}\P_4 e^{-B}\label{hred}
\end{align}
where the various terms are defined below.
$\P_{n}$ indicate polynomials of degree $n$ in $a, a^*$ to be given explicitly:
\begin{align}
\P_{0}&:=\int dx\left\{\frac{1}{2i}\big(\phi\bar{\phi}_{t}-\bar{\phi}\phi_{t}\big)-\big\vert\nabla\phi\big\vert^{2}
\right\}\nonumber
\\
&-\frac{1}{2}\int dxdy\left\{v_{N}(x-y)\vert\phi(x)\vert^{2}\vert\phi(y)\vert^{2}\right\}\ . \label{phase1-I}\\
\P_{1}&:=\int dx\left\{ h(t,x)a^{\ast}_{x}+\bar{h}(t,x)a_{x}\right\}\\
=&a^*(h(t, \cdot) + a(\overline h(t, \cdot)) \notag
\end{align}
where $h:=-(1/i)\partial_{t}\phi +\Delta\phi -\big(v_{N}\ast\vert\phi\vert^{2}\big)\phi$.
\begin{subequations}
\begin{align}
&\H_{G}:=\frac{1}{2}\int dxdy\left\{-g_{N}(t,x,y)a_y^*a_x-g_{N}(t,y,x)a_x^*a_y\right\}\\
&\mbox{where}\\
&g_{N}(t, x, y):= -\Delta_x \delta (x-y)
 +(v_N * |\phi|^2 )(t, x) \delta(x-y) \nonumber \\
&\qquad\qquad\  +v_N(x-y) \overline\phi (t, x)  \phi(t, y)\label{op-g}
\end{align}
\end{subequations}
and
\begin{align}
&R=\frac{1}{i}\left(\frac{\partial}{\partial t}e^{K}\right)e^{-K} + [G, e^{K}]e^{-K}+ e^{K} M e^{-K}= \notag \\
&= \left(
\begin{matrix} \label{rformula}
-\overline{\W(\chb)} & -\overline{\S(\sh)}\\
\S(\sh) & \W(\chb)
\end{matrix}
\right)\circ\left(
\begin{matrix}
\ch &-\shb\\
-\sh & \chb
\end{matrix}
\right)\\
&+e^{K} M e^{-K} \notag
\end{align}
where
 ${\bf S}$ describes a Schr\"odinger type evolution, while ${\bf W}$ is a Wigner type operator by
\begin{align*}
&{\bf S}(s):= \frac{1}{i} s_t  + g^T \circ s + s \circ g\quad {\rm and}\quad
{\bf W}(p):= \frac{1}{i} p_t  +[g^T, p]\\
& {\rm while}\quad M :=
\left(
\begin{matrix}
0 &\overline{m}\\
-m & 0
\end{matrix}
\right) \quad {\rm where}\quad\\
&m(x, y) :=- v_N(x-y) \phi(x) \phi(y), \, \,
v_N(x)= N^{3 \beta} v(N^{\beta}x)\\
& {\rm and}\quad G :=
\left(
\begin{matrix}
g & 0\\
0 & -g^T
\end{matrix}
\right).
\end{align*}
Finally,
\begin{subequations}
\begin{align}
\P_{3}&: =[\A, \V]=\int dxdy\left\{v_{N}(x-y)\big(\phi(y)a^{\ast}_{x}a^{\ast}_{y}a_{x}+\bar{\phi}(y)a^{\ast}_{x}a_{x}a_{y}\big)\right\}
\label{cubic-1}
\\
\P_{4}&:= \V= (1/2)\int dxdy\left\{v_{N}(x-y)a^{\ast}_{x}a^{\ast}_{y}a_{x}a_{y}\right\}\ . \label{quartic-1}
\end{align}
\end{subequations}
The main result of \cite{GM}, building on the previous papers of the authors and D. Margetis \cite{GMM1, GMM2},  can be summarized as follows.
\begin{theorem}\label{GMthm} Let $\phi$ and $k$ satisfy
\begin{subequations}
\begin{align}
&\frac{1}{i}\partial_{t}\phi -\Delta\phi +\big(v_{N}\ast\vert\phi\vert^{2}\big)\phi =0 \label{meanfield-3}
\\
&\mbox{and either one of the following equivalent equations:}\\
& 1) \, \left({\bf S}\left(\sh\right) - \chb \circ m \right)\circ \ch = \left({\bf W}\left(\chb\right) + \sh \circ \overline m \right) \circ \sh \\
&\mbox{ or else the equivalent non-liner equation} \\
& 2) \, \S (\th) = m + \th \circ \overline m \circ \th \\
&\mbox{where } \, \th:= \chb^{-1} \circ \sh\\
&\mbox{or else the equivalent system of liner equations}\\
& 3a) \, {\bf S}\left(\sht\right)=m_{N} \circ \cht + \chbt \circ m_{N}  \label{pair-S-1} \\
& 3b) \, {\bf W}\left(\chbt\right)= m_{N} \circ \shbt- \sht \circ \overline m_{N}\ . \label{pair-W-1}
\end{align}
\end{subequations}
with prescribed initial conditions $\phi(0, \cdot)= \phi_0$, $k(0, \cdot, \cdot)=0$. If $\phi$, $k$ satisfy the above equations, then there exists a real phase function $\chi$ such that
\begin{equation}\label{mainestim}
\big\Vert\big\vert\psi_{exact}(t)\big>-
e^{iN\chi(t)}\big\vert\psi_{appr}(t)\big>\big\Vert_{\F}\leq \frac{C (1+t) \log^4(1+t)}{N^{(1-3\beta)/2}}\ .
\end{equation}
provided $0 < \beta < \frac{1}{3}$.
\end{theorem}

The purpose of the present paper is to introduce and study a coupled refinement of the system \eqref{meanfield-3}, \eqref{pair-S-1},
\eqref{pair-W-1} which, we believe,  is the correct system describing the case $\beta=1$. These equations occur as Euler-Lagrange equations, and are written down explicitly in Theorem \eqref{mainthm}.

\section{Main new results}

Since $\H_{red}$ is a fourth order polynomial in $a$ and $a^*$,
\begin{align}
\H_{red} \vac =(X_0, X_1, X_2, X_3, X_4, 0, \cdots). \label{Xidef}
\end{align}
\begin{definition}
Define the Lagrangian
\begin{align}
\mathcal{L}= -\int X_0(t) dt \label{lagrangian}
\end{align}
\end{definition}
 The new, coupled equations for $\phi$ and $k$ that we introduce in this paper are  $X_1=0$ and $X_2=0$.

 We first prove that $\L$ is indeed the Lagrangian for these equations. We start by showing "abstractly" that
 \begin{align}
 &\frac{\delta X_0}{\delta \overline \phi}= \sqrt{N} \int \left( X_1(t,  x) \ch(t,  x, \cdot) - \overline X_1(t,  x) \sh(t, x, y)
 \right) dx \label{phivar1}\\
 &\frac{\delta X_0}{\delta \overline \zeta}=
 \frac{1}{\sqrt 2 } \chb\circ X_2 \circ \ch \label{zetavar1}
 \end{align}
where $\z=\th=\chb^{-1}\circ \sh$.
We then compute explicitly the zeroth order term  $X_0(t)$ in $\H_{red} \vac $ (which provides the Lagrangian density for our coupled equations):
\begin{align*}
& - X_0(t)
=
N \int
  dx_1\left\{-\Im \left(\phi_1 \overline{\partial_t \phi_1} \right)+\big\vert\nabla\phi_1\big\vert^{2}
\right\}
\\
&+\frac{N}{2}\int  dx_1dx_2 v^{N}_{1-2}\vert\phi_1 \phi_2 + \frac{1}{N} (\shh \circ \chh)_{1, 2}\vert^{2}\\
&+\frac{1}{2}\int  dx_1dx_2 dx_3v^{N}_{1-2}\vert \phi_1 \shh_{2, 3} + \phi_2 \shh_{1, 3} \vert^2\\
+&\frac{1}{2} \Bigg(\int  dx_1 dx_2\left\{- \Im\left(\shh_{1, 2} \overline{\partial_t \shh_{1, 2}} \right) + \big \vert \nabla_{1, 2} \shh_{1, 2}\big \vert^2 \right \}\\
+&\frac{1}{2N} \int dt dx_1 dx_2
 v^N_{1-2}\Big\{
|(\shh \circ \overline \shh)_{1, 2}|^2 + (\shh \circ \overline \shh)_{1, 1}(\overline \shh \circ  \shh)_{2, 2}
\Big\}
\Bigg) \ .
\end{align*}
where $\shh_{1, 2}$ is an abbreviation for $\sh(t, x_1, x_2)$, $v^{N}_{1-2}=v_N(x_1-x_2)$, etc, and the products are pointwise products, while compositions
are denoted by $\circ$.
Then we proceed to compute explicitly the coupled equations $X_1=0$ and $X_2=0$, derive conserved quantities, and formulate a conjecture. The resulting equations are similar to those of Theorem \eqref{GMthm}, except that
$m=-v_N(x_1-x_2) \phi(t, x_1)\phi(t, x_2)$ is replaced by
\begin{align*}
\Theta=-v_N(x_1-x_2)\left(\phi(t, x_1)\phi(t, x_2) + \frac{1}{2N}  \sht(t, x_1, x_2)\right),
\end{align*}
and similar $O (\frac{1}{N})$ coupling corrections apply to the Hartree operator as well as $\S$ and $\W$.

\begin{remark} The static terms of $X_0(t)$ (not involving time derivatives) also appear in the recent preprint
\cite{B-O-S}, but do not serve as a Lagrangian there.
\end{remark}

\section{The Lagrangian and the equations, abstract formulation}

\begin{proposition} \label{phivar}
Let $k$ and  $\phi$ be fixed.
\begin{align}
 &\frac{d}{d \epsilon}\Big|_{\epsilon =0}X_0(\phi + \epsilon  h, k) \label{vark}\\
 &= 2 \sqrt{N} \Re \int X_1(t,  x) \left(\ch(t,  x, y) \overline h (t, y) -\shb(t, x, y) h(t, y)\right) dx dy \notag
 \end{align}
 In particular, if this vanishes for all $h$, then $X_1(t, x)=0$.
 \end{proposition}

 \begin{proof}
$\H_{red}$ can be written  as
\begin{align*}
\H_{red}=\frac{1}{i}\frac{\partial}{\partial t} + e^{\B}e^{\sqrt N \A}\left(-\frac{1}{i}\frac{\partial}{\partial t} + \H\right)\circ e^{-\sqrt N \A}e^{-\B}
\end{align*}
in the sense of compositions  (in space only) of operators.
During this proof, denote $\H_t=-\frac{1}{i}\frac{\partial}{\partial t} + \H$.

Let $h$ be an $L^2$ function and let\\
 $\A_{\epsilon}=\sqrt N (a\left(\overline{\phi} + \epsilon \overline h) -a^*(\phi + \epsilon  h)\right)$.
Thus we have
\begin{align*}
X_0(\phi + \epsilon  h, k)=\bigg<e^{\B}e^{\A_{\epsilon}} H_t e^{-\A_{\epsilon}}e^{-\B} \vac, \vac \bigg>
\end{align*}
We compute
\begin{align*}
&\left(\frac{d}{d \epsilon}\Big|_{\epsilon =0} e^{\A_{\epsilon}}\right) e^{-\A_{0}}=
\dot{\A_0} + \frac{1}{2}[\A_0, \dot {\A_0}]\\
&=\sqrt N \left(a(\overline{ h})-a^*(h)\right) + \frac{N}{2} [a(\overline{ \phi})-a^*(\phi), a(\overline h)-a^*(h)]\\
&=\sqrt N \left(a(\overline{ h})-a^*(h)\right) + iN \Im \int \phi \overline h
\end{align*}
and
\begin{align*}
e^{\A_{0}}\left(\frac{d}{d \epsilon}\Big|_{\epsilon =0} e^{-\A_{\epsilon}}\right) =-\left(\frac{d}{d \epsilon}\Big|_{\epsilon =0} e^{\A_{\epsilon}}\right) e^{-\A_{0}}
\end{align*}
thus
\begin{align*}
&\frac{d}{d \epsilon}\Big|_{\epsilon =0} \bigg<e^{\B}e^{\A_{\epsilon}} \H_t e^{-\A_{\epsilon}}e^{-\B} \vac, \vac \bigg>\\
&=\bigg< e^{\B}\Big[\sqrt N a(\overline h)-\sqrt N a^*( h), e^{\A_0}\H_t e^{-\A_0}\Big] e^{-\B}\vac, \vac\bigg>\\
&=\bigg< \Big[e^{\B}\left(\sqrt N a(\overline h) - \sqrt N a^*( h)\right)e^{-\B}, e^{\B} e^{\A_0}\H_t e^{-\A_0}e^{-\B}\Big] \vac, \vac\bigg>\\
&=\bigg< \Big[a(\overline l)-a^*(l), e^{\B} e^{\A_0}\H_t e^{-\A_0}e^{-\B}\Big] \vac, \vac\bigg>\\
& = 2 \Re \bigg<  \H_{red} \vac, a^*(l)\vac\bigg> :=I
\end{align*}
where we denoted
\begin{align*}
e^{\B}\left(\sqrt N a(\overline h)-\sqrt N a^*( h)\right)e^{-\B}=a(\overline l)-a^*(l)
\end{align*}
Explicitly,
\begin{align*}
&e^{\B}\left( a(\overline h)- a^*( h)\right)e^{-\B}\\
&=a(\ch \circ \overline h ) + a^*(\sh \circ \overline h )\\
&-a(\shb \circ h) - a^*(\chb \circ h)
\end{align*}
so
\begin{align*}
l= \sqrt{N}\left(\chb \circ h- \sh \circ \overline h\right)
\end{align*}
Thus,
\begin{align*}
I= &2 \sqrt{N}\Re \int X_1(t, x) \left(\ch(t, x, y) \overline h ( y) -\shb(t, x, y) h( y)\right) dx dy\\
=& 2 \sqrt{N}\Re \int \left(\chb \circ X_1 - \sh \circ \overline{X_1} \right)(y)\overline h(y) dy
\end{align*}
\end{proof}
In order to state the corresponding result for $X_2$, we have to introduce a new set of coordinates for our basic matrices
\begin{align*}
e^K= \left(
\begin{matrix}
\ch & \shb\\
\sh & \chb
\end{matrix}
\right)
\end{align*}
where
\begin{align}
K=\left(\begin{matrix}
0 & \overline{k}(t, x, y)\\ \label{K2}
k(t, x, y)& 0
\end{matrix}
\right)
\end{align}
The most obvious coordinate system is, of course, provided by $k$. We recall the following proposition, proved in \cite{GMM2}.
\begin{proposition} \label{exp} The exponential map is one-to-one and onto from matrices of the form
\eqref{K2} ($k \in L^2$, symmetric) to positive definite matrices $E$ satisfying the three properties of Remark
\eqref{E} for which $\|I-E\|_{L^2}$ is finite.
\end{proposition}
For our purposes, a better coordinate system is provided by $\zeta= \th=\chb^{-1} \circ \sh$.
\begin{proposition} There is a bijection between $k \in L^2$, symmetric, and $\zeta \in L^2$, symmetric, $\|\zeta\|_{op} <1$ ($op$ stands for the operator norm) such that
\begin{align}
e^K:=&\left(
\begin{matrix}
\ch &\shb\\
\sh & \chb
\end{matrix}
\right) \notag\\
=E_{\zeta}:=&
\left(
\begin{matrix}
I &\overline \zeta\\
0 & I
\end{matrix}
\right) \label{ez}
\left(
\begin{matrix}
(I - \overline {\zeta} \circ \zeta)^{1/2}&0\\
0 & (I -  \zeta \circ \overline{\zeta})^{-1/2}
\end{matrix}
\right)
\left(
\begin{matrix}
I &0\\
\zeta & I
\end{matrix}
\right)\\
=&\left(
\begin{matrix}
(I- \overline {\zeta} \circ \zeta)^{-1/2} & \overline{\zeta} \circ (I-  {\zeta} \circ \overline\zeta)^{-1/2}\\
\zeta \circ (I- \overline {\zeta} \circ \zeta)^{-1/2} & (I-  {\zeta} \circ \overline \zeta)^{-1/2}
\end{matrix}
\right) \notag
\end{align}
where the square root is taken in the operator sense.
\end{proposition}
\begin{proof} Given $k$, define $\zeta= \chb^{-1} \circ \sh$. The decomposition \eqref{ez} is an algebraic identity, and it is clear that $\zeta $ is symmetric and $L^2$.
 Since
  $I - \ch^{-2} =  \overline {\zeta} \circ \zeta$, we see that  $\|\zeta\|_{op}<1$.
  In fact, $\|\zeta  v\|^2_{L^2} = \|v\|^2_{L^2}-\|\ch^{-1} v\|^2_{L^2}$.
Conversely, given $\zeta$ a
 symmetric Hilbert-Schmidt kernel with
$\|\zeta\|_{op} <1$ define $E_{\zeta}$ by \eqref{ez}. It is easy to check that
$E_{\zeta}$ is positive definite, satisfies the symmetries  of remark \eqref{E} and $\|I - E_{\zeta}\|_{HS} < \infty$.
($HS$ stands for the Hilbert-Schmidt norm), thus we can apply Proposition \eqref{exp} and find the corresponding $K$.
\end{proof}
We also record the following consequence:
\begin{proposition} \label{kepsilon} Let $\zeta_{\epsilon}= \zeta + \epsilon h$ ($h \in L^2$, symmetric, $\|\zeta_{\epsilon}\|_{op}<1$), and $K_{\epsilon}$ corresponding to $\zeta_{\epsilon}$ according to the previous proposition. Then
\begin{align*}
\frac{d}{d \epsilon}\Big|_{\epsilon =0}e^{K_{\epsilon}}e^{-K}=
\left(
\begin{matrix}
ia  &\overline{b}\\
b & -i a^T
\end{matrix}
\right)
\end{align*}
\end{proposition}
with
\begin{align*}
 b
=\chb \circ h \circ \ch
\end{align*}
\begin{proof}
We compute
\begin{align*}
&\frac{d}{d \epsilon}\Big|_{\epsilon =0}e^{K_{\epsilon}}e^{-K}\\
&=
\left(
\begin{matrix}
\ch' \circ \ch - \shb'\circ \sh & -\ch'\circ\shb + \shb' \circ \chb\\
\sh' \circ \ch- \chb' \circ \sh & -\sh' \circ \shb + \chb'\circ \chb
\end{matrix}
\right)
\end{align*}
An easy calculation shows that $ b=-\chb'\circ\sh + \sh' \circ \ch
=\chb \circ \z' \circ \ch$.
\end{proof}
We are ready to prove
 \begin{align*}
 \frac{\delta X_0}{\delta \overline \zeta}=
 \frac{1}{\sqrt 2 } \chb\circ X_2 \circ \ch
 \end{align*}
\begin{proposition} \label{vark}
Let $k_{\epsilon}$ correspond to $\zeta + \epsilon h$ as in the previous proposition. Then
\begin{align*}
 \frac{d}{d \epsilon}\Big|_{\epsilon =0}X_0(\phi , k_{  \epsilon } )=
 \sqrt 2 \Re \int \chb\circ X_2 \circ \ch (t, z, w) \overline h (t, z, w) dz dw
 \end{align*}
 In particular, if the above vanishes for all $h$, then $X_2=0$.
 \end{proposition}
\begin{proof}
Let $B_{\epsilon}= B(k_{\epsilon} )$.
\begin{align*}
X_0(\phi , k_{ \epsilon })=\bigg<e^{\B_{\epsilon}}e^{\sqrt N \A} H_t e^{-\sqrt N \A}e^{-\B_{\epsilon}} \vac, \vac \bigg>
\end{align*}
and
\begin{align}
\frac{d}{d \epsilon}\Big|_{\epsilon =0}X_0(\phi , k_{\epsilon })=
-2 \Re
\bigg<H_{red} \vac, \psi\vac \bigg> \label{real}
\end{align}
where
\begin{align*}
\psi=\frac{d}{d \epsilon}\Big|_{\epsilon =0}e^{\B_{\epsilon}} e^{-B}=
\I\left(\frac{d}{d \epsilon}\Big|_{\epsilon =0}e^{K_{\epsilon}}e^{-K}\right)
\end{align*}
Using the isomorphism \eqref{liemap} and proposition \eqref{kepsilon} we see that
\begin{align*}
\psi \vac  = (i \theta, 0, -\frac{1}{\sqrt 2}\chb \circ h \circ \ch (t, x_1, x_2), 0, \cdots)
\end{align*}
where $\theta$ is a real number coming from the trace of the self-adjoint $a$.
Since $X_0$ is real, $i \theta$ does not  contribute to \eqref{real}, and the result follows.
\end{proof}

\section{Explicit form of the Lagrangian}
The goal of this section is the following proposition.
\begin{proposition} \label{propL}The zeroth order term in $\H_{red} \vac $ (which provides the Lagrangian density for our coupled equations)
is $X_0(t)$ where

\begin{align*}
& - X_0(t)
=
N \int
  dx_1\left\{-\Im \left(\phi_1 \overline{\partial_t \phi_1} \right)+\big\vert\nabla\phi_1\big\vert^{2}
\right\}
\\
&+\frac{N}{2}\int  dx_1dx_2 v^{N}_{1-2}\vert\phi_1 \phi_2 + \frac{1}{N} (\shh \circ \chh)_{1, 2}\vert^{2}\\
&+\frac{1}{2}\int  dx_1dx_2 dx_3v^{N}_{1-2}\vert \phi_1 \shh_{2, 3} + \phi_2 \shh_{1, 3} \vert^2\\
+&\frac{1}{2} \Bigg(\int  dx_1 dx_2\left\{- \Im\left(\shh_{1, 2} \overline{\partial_t \shh_{1, 2}} \right) + \big \vert \nabla_{1, 2} \shh_{1, 2}\big \vert^2 \right \}\\
+&\frac{1}{2N} \int dt dx_1 dx_2
 v^N_{1-2}\Big\{
|(\shh \circ \overline \shh)_{1, 2}|^2 + (\shh \circ \overline \shh)_{1, 1}(\overline \shh \circ  \shh)_{2, 2}
\Big\}
\Bigg) \ .
\end{align*}
where $\shh_{1, 2}$ is an abbreviation for $\sh(t, x_1, x_2)$, etc, and the products are pointwise products, while compositions
are denoted by $\circ$.
\end{proposition}
The proof follows from several lemmas, which can be proved by explicit calculations.
We proceed to compute $X_0$ in \eqref{Xidef}. The only terms in \eqref{hred} which contribute to $X_0$ are $N \P_0$ which is already explicit,  the zeroth order terms in  $\I(R) \vac$, as well as the zeroth order terms in $N^{-1}e^{\B}\P_4 e^{-B} \vac$.
\begin{lemma} The term $N \P_0$ is given by
\begin{align*}
N \P_{0}&=N\int dx\left\{\frac{1}{2i}\big(\phi\bar{\phi}_{t}-\bar{\phi}\phi_{t}\big)-\big\vert\nabla\phi\big\vert^{2}
\right\}
\\
&-\frac{N}{2}\int dx_1dx_2\left\{v^{N}_{1-2}\vert\phi_1 \phi_2\vert^{2}\right\}\ .
\end{align*}
\end{lemma}
We  used abbreviations $v^{N}_{1-2}=v_N(x_1-x_2)$, $\phi_1=\phi(x_1)$, etc., and for the following two lemmas
we will denote $u_{1, 2}=\sh(t, x_1, x_2)$ and $c_{1, 2}=\ch(t, x_1, x_2)$.
\begin{lemma}
The zeroth order term in $\I(R) \vac$ is
\begin{align*}
&-\frac{1}{2} \Bigg(\int dx_1 dx_2\left\{\frac{1}{2i}\left( \bar{u}_{1, 2}\partial_t u_{1, 2}-\partial_t \bar{u}_{1, 2} u_{1, 2}\right) + \big \vert \nabla_{1, 2} u_{1, 2}\big \vert^2 \right \}\\
&+\int dx_1 dx_2 dx_3 \left\{ v^N_{1-2} |\phi_1 u_{2, 3}|^2 + |\phi_2 u_{1, 3}|^2\right \} \\
&+ 2 \Re\int  dx_1 dx_2  dx_3 \left\{ v^N_{1-2} \phi_2 u_{1, 3} \overline{ \phi_1  u_{2, 3} }\right\}\\
&+ 2 \Re \int  dx_1 dx_2  \left\{v^N_{1-2}(u \circ c)_{1, 2} \bar \phi_1 \bar \phi_2 \right\}
\Bigg)
\end{align*}
\end{lemma}

\begin{lemma}
The zeroth order term in $-\frac{1}{N}e^{\B} \V e^{-\B} \vac$ is
\begin{align*}
&-\frac{1}{2}\int dx_{1}dx_{2} v^N_{1-2}\Big\{
(u \circ c)_{1, 2}\overline{(u \circ c)_{1, 2}}\\
&+|(u \circ \overline u)_{1, 2}|^2 + (u \circ \bar u)_{1, 1}(\bar u \circ  u)_{2, 2}
\Big\}\ .
\end{align*}
\end{lemma}

\section{Explicit form of the equations}
In this section we derive the following theorem, thus introducing our new equations.
First, some notation.
Consider the kernels
\begin{align*}
&\omega_c(t, x, y)= \overline{\phi}(t, x)\phi(t, y)\\
&\omega_p(t, x, y)= \shb \circ \sh (t, x, y)
\end{align*}
and their trace densities
\begin{align*}
&\rho_c=|\phi|^2(t, x)\\
&\rho_p(t, x)= \sh \circ \shb (t, x, x)
\end{align*}
Here $c$ stands for condensate, and $p$ for pair. In this notation, the old operator kernel $g_N$ defined in \eqref{op-g} is
\begin{align*}
&g_{N}(t, x, y):= -\Delta_x \delta (x-y)
 +(v_N * \rho_c )(t, x) \delta(x-y)  \\
&\qquad\qquad\  +v_N(x-y) \omega_c(t, x, y)
\end{align*}
Define the new operator kernel
\begin{align}
&\tilde g_{N}(t, x, y):= -\Delta_x \delta (x-y) \notag\\
&+(v_N * \rho_c )(t, x) \delta(x-y) +v_N(x-y) \omega_c(t, x, y) \label{alphac}\\
&+\frac{1}{N}\left((v_N * \rho_p )(t, x) \delta(x-y) +v_N(x-y) \omega_p(t, x, y)\right)\label{alphap}
\end{align}
Also denote $\alpha_c=\eqref{alphac}$, $\frac{1}{N}\alpha_p=\eqref{alphap}$ and $\alpha=\alpha_c+\frac{1}{N}\alpha_p$.
Define
\begin{align*}
&\tilde{\bf S}(s):= \frac{1}{i} s_t  + \tilde g_{N} ^T \circ s + s \circ \tilde g_{N}\quad {\rm and}\quad
\tilde{\bf W}(p):= \frac{1}{i} p_t  +[\tilde g_{N}^T, p]
\end{align*}
Finally, define
 $\Theta(t, x_1, x_2)=- v_N(t, x_1, x_2)\left( \phi(t, x_1)\phi(t, x_2) + \frac{1}{2N}  \sht(t, x_1, x_2)\right)$.

\begin{theorem} \label{mainthm}
The equation $X_1=0$ is equivalent to
\begin{align*}
&\frac{1}{i}\partial_{t}\phi(t, x_1) -\Delta\phi -\int \Theta(t, x_1, x_2) \overline \phi(t, x_2) dx_2
+ \int \frac{1}{N} \alpha_p^T(t, x_1, x_2)\phi(t, x_2) dx_2 =0
\end{align*}
The equation $X_2=0$ is equivalent to either of :

1) the equation
\begin{align*}
 \tilde{\S}(\th)
=  \Theta   +   \th \circ \overline { \Theta} \circ \th
\end{align*}
2) the pair of equations (in fact, 2a) implies 2b))
\begin{align}
&2a) \, \tilde{\S}\left(\sht\right)  =\Theta \circ \cht + \chbt \circ  \Theta\\
&2b) \, \tilde{\W}\left(\chbt\right)= \Theta \circ \shbt- \sht \circ \overline { \Theta} \notag
\end{align}
\end{theorem}
\begin{remark}
One can go back and fourth between $\z$ and $\cht, \sht$ using
\begin{align*}
&\shb \circ \sh = (1- \zb \circ \z)^{-1} -1 =\frac{1}{2}\left(\cht -1\right)\\
&\z=\sht(1+\cht)^{-1}
\end{align*}
\end{remark}
\begin{proof}
A direct calculation for $X_1$  shows that
\begin{align*}
X_1 = -\sqrt N \left(\chb \circ \widetilde{Har}_k(\phi) + \sh \circ \overline{\widetilde{Har}_k(\phi)}\right)
\end{align*}
where
\begin{align*}
&\widetilde{Har}_k(\phi)(t, x_1)\\
&=\frac{1}{i}\partial_{t}\phi -\Delta\phi -\int \Theta(t, x_1, x_2) \overline \phi(t, x_2) dx_2\\
&+ \frac{1}{N} \int v_{N}(x_{1}-x_{2})
(\shh\circ\shhb)(x_{1},x_{2})\phi(x_{2}) dx_2\\
&+\frac{1}{N}\phi(x_{1})  \int v_{N}(x_{1}-x_{2})(\shh\circ\shhb)(x_{2},x_{2}) dx_2
\end{align*}
In conjunction with Proposition \eqref{phivar} this shows that
\begin{align*}
\frac{\delta \L}{\delta \overline \phi} = N \widetilde{Har}_k(\phi)
\end{align*}
which can also be easily verified directly from Proposition \eqref{propL}.

A direct calculation also shows that, if $X_2$ denotes
the second component of $\H_{red} \vac$, then
\begin{align}
&-\sqrt 2X_2 (t, y_1, y_2) =\label{X2explicit}\\
&\bigg( \left(\S(\sh) - \chb \circ m \right) \circ \ch - \left(\W(\chb) + \sh \circ \overline m \right) \circ \sh\bigg) \notag\\
+&(1/N)\int dx_{1}dx_{2}\qquad\Big\{  \notag\\
&\bigg(\chhb(y_{1},x_{2})\shh(x_{2},y_{2})\big(\shhb\circ\shh\big)(x_{1},x_{1})v_{N}(x_{1}-x_{2}) + \notag
\\
&\chhb(y_{1},x_{2})\shh(x_{1},y_{2})\big(\shhb\circ\shh\big)(x_{1},x_{2})v_{N}(x_{1}-x_{2})  + \notag
\\
&\chhb(y_{1},x_{1})\shh(x_{2},y_{2})\big(\shh\circ\shhb\big)(x_{1},x_{2})v_{N}(x_{1}-x_{2})
 + \notag
\\
&\chhb(y_{1},x_{1})\shh(x_{1},y_{2})\big(\shh\circ\shhb\big)(x_{2},x_{2})v_{N}(x_{1}-x_{2})\bigg)_{symm} + \notag
\\
&\shh(y_{1},x_{1})\shh(x_{2},y_{2})\big(\shhb\circ\chhb\big)(x_{1},x_{2})v_{N}(x_{1}-x_{2}) + \notag
\\
&\chhb(y_{1},x_{1})\chh(x_{2},y_{2})\big(\chhb\circ\shh\big)(x_{1},x_{2})v_{N}(x_{1}-x_{2}) \notag
\Big\}\ .
\end{align}
where $symm$ stands for "symmetrized". The time dependance in the last six lines has been omitted.
Recalling $\zeta = \chb^{-1} \circ \sh = \sh \circ \chb^{-1}$, compose on the left with $\chb^{-1}$ and on the right with $\ch^{-1}$
to get
\begin{align}
&\chb^{-1}\circ X_2 \circ \ch^{-1}
= \S(\zeta) -   \Theta   -   \z \circ \overline{ \Theta}  \circ \z + \frac{1}{N} \N
\end{align}
where  $N$ is given by
\begin{align*}
\N(t, y_1, y_2)=& \z(t, y_{1},y_{2})
\bigg(\int dx
\bigg(\big(\shhb\circ\shh + \shh\circ\shhb\big)(t, x, x)v_{N}(x -y_1)\bigg)_{symm} +
\\
&\bigg(\int dx
\z(t, x, y_{2})\big(\shhb\circ\shh + \shh\circ\shhb\big)(t, x, y_1)v_{N}(x -y_1) \bigg)_{symm}
\end{align*}
where $symm$ stands for symmetrizing in $y_1, y_2$. In other words,
\begin{align*}
\N= \z \circ \alpha_p + \alpha_p \circ \z
\end{align*}
Thus, in $\z$ coordinates, the equation $X_2=0$ becomes
\begin{align}
\tilde{\S}(\zeta) -\Theta   -   \z \circ \overline {\Theta} \circ \z =0\label{thetaeq}
\end{align}
Now we can get an equation for $\tilde{\W}(\chbt)$ and $\tilde{\S}(\sht)$.
We will use the general formulas
\begin{align*}
&\tilde{\W}(f^{-1})=-f^{-1}\circ \tilde{\W}(f) \circ f^{-1}\\
& \tilde{\W}(f \circ \overline g)=\tilde{\S}(f) \circ \overline g- f \circ \overline{\tilde{\S}(g)}\\
&\tilde{\S}(f \circ g)=\tilde{\S}(f) \circ g - f \circ \overline{\tilde{\W} \left(\overline g \right)}
\end{align*}
Thus
\begin{align*}
&\tilde{\W}\left((1- \z \circ \zb)^{-1}\right)=
(1- \z \circ \zb)^{-1} \circ \left(\tilde{\S}(\z) \circ \zb -\z \circ \overline{\tilde{\S}(\z)}\right)\circ(1- \z \circ \zb)^{-1}\\
&=(1- \z \circ \zb)^{-1}\circ \Bigg(\left(  \Theta   +  \z \circ \overline{ \Theta}  \circ \z \right) \zb
+\z \circ \overline {\left(  \Theta   +  \z \circ \overline{ \Theta}  \circ \z \right)}
\Bigg)(1- \z \circ \zb)^{-1}\\
\end{align*}
Similarly we get a formula for $\tilde{\S}(\sht)$, using
\begin{align*}
&\tilde{\S}\left(\z \circ (1- \zb \circ \z)^{-1}\right)\\
&=
(1- \z \circ \zb)^{-1}\circ \left( \tilde{\S}(\z) - \z \circ \overline{\tilde{\S}(\z)} \circ \z \right) \circ (1- \zb \circ \z)^{-1}
\end{align*}
\begin{align*}
&\tilde{\S}\left(\z \circ (1- \zb \circ \z)^{-1}\right)=\\
&(1- \z \circ \zb)^{-1}\circ \Bigg( \Theta   +   \z \circ \overline {\Theta} \circ \z
+\z \circ \overline{\left(
 \Theta   +   \z \circ \overline {\Theta} \circ \z \right)}
 \circ \z \Bigg) \circ (1- \zb \circ \z)^{-1}\\
&= \left((1- \z \circ \zb)^{-1} - \frac{1}{2}\right) \circ  \Theta  + \Theta  \circ \left((1- \zb \circ \z)^{-1} - \frac{1}{2}\right)
\end{align*}
\end{proof}

\section{Conserved quantities}
We start by motivating the introduction of some conserved quantities. Recall the Lagrangian
\begin{align*}
&\L (\phi, \sh)
=
N \int
 dt dx_1\left\{-\Im \left(\phi_1 \overline{\partial_t \phi_1} \right)+\big\vert\nabla\phi_1\big\vert^{2}
\right\}
\\
&+\frac{N}{2}\int dt dx_1dx_2 v^{N}_{1-2}\vert\phi_1 \phi_2 + \frac{1}{N} (\shh \circ \chh)_{1, 2}\vert^{2}\\
&+\frac{1}{2}\int dt dx_1dx_2 dx_3v^{N}_{1-2}\vert \phi_1 \shh_{2, 3} + \phi_2 \shh_{1, 3} \vert^2\\
+&\frac{1}{2} \Bigg(\int dt dx_1 dx_2\left\{- \Im\left(\shh_{1, 2} \overline{\partial_t \shh_{1, 2}} \right) + \big \vert \nabla_{1, 2} \shh_{1, 2}\big \vert^2 \right \}\\
+&\frac{1}{2N} \int dt dx_1 dx_2
 v^N_{1-2}\Big\{
|(\shh \circ \overline \shh)_{1, 2}|^2 + (\shh \circ \overline \shh)_{1, 1}(\overline \shh \circ  \shh)_{2, 2}
\Big\}
\Bigg) \ .
\end{align*}
where $\shh_{1, 2}$ is an abbreviation for $\sh(t, x_1, x_2)$, etc, and the products are pointwise products, while compositions
are denoted by $\circ$.
 Introduce the energy $\E$
\begin{align*}
&\E (\phi, \sh)(t)
=
N \int
  dx_1\left\{\big\vert\nabla\phi_1\big\vert^{2}
\right\}
\\
&+\frac{N}{2}\int dx_1dx_2 v^{N}_{1-2}\vert\phi_1 \phi_2 + \frac{1}{N} (\shh \circ \chh)_{1, 2}\vert^{2}\\
&+\frac{1}{2}\int  dx_1dx_2 dx_3v^{N}_{1-2}\vert \phi_1 \shh_{2, 3} + \phi_2 \shh_{1, 3} \vert^2\\
+&\frac{1}{2} \Bigg(\int  dx_1 dx_2\left\{ \big \vert \nabla_{1, 2} \shh_{1, 2}\big \vert^2 \right \}\\
+&\frac{1}{2N} \int  dx_1 dx_2
 v^N_{1-2}\Big\{
|(\shh \circ \overline \shh)_{1, 2}|^2 + (\shh \circ \overline \shh)_{1, 1}(\overline \shh \circ  \shh)_{2, 2}
\Big\}
\Bigg) \ .
\end{align*}
Our equations for $\phi$ and $\sh$ are equivalent to
\begin{align}
&N \frac{1}{i} \frac{\partial \phi}{\partial t} = - \frac{\delta \E}{\delta \overline \phi} \label{varphi}\\
&\frac{1}{i} \frac{\partial \sh}{\partial t} = - \frac{\delta \E}{\delta \overline \sh} \label{vark}
\end{align}
The relation
\begin{align*}
 0=&\frac{d}{d \theta}\big|_{\theta=0} \E(e^{i \theta} \phi, e^{2 i \theta} \sh)\\
 =&2 \Re \left(\int \frac{\delta \E}{\delta \overline \phi}(- i \overline \phi)dx_1+ \int\frac{\delta \E}{\delta \overline \sh} (
 -i \overline \sh)dx_1 dx_2 \right)
 \end{align*}
together with
\eqref{varphi}, \eqref{vark}, leads to the conservation
\begin{align*}
 \frac{d}{dt}\left(\int|\phi(t, x_1)|^2 dx_1 + \frac{1}{N} \int |\sh(t, x_1, x_2)|^2 dx_1 dx_2\right) =0
 \end{align*}
 thus we define the density
 \begin{align*}
 \rho(t, x_1) &= |\phi(t, x_1)|^2  + \frac{1}{N} \int |\sh(t, x_1, x_2)|^2  dx_2\\
 &=\rho_c(t, x_1) + \frac{1}{N}\rho_p(t, x_1)
 \end{align*}

 Similarly, let  $\phi_{\epsilon}(t, x)=
 \phi(t, x+ \epsilon e_j)$, $\sh_{\epsilon}(t, x, y)=
 \sh(t, x+ \epsilon e_j, y+ \epsilon e_j)$ ($e_j=$ unit vector, $1 \le j \le 3$). The relation
 \begin{align*}
 0=&\frac{d}{d \epsilon}\Big|_{\epsilon=0} \E(\phi_{\epsilon}, \sh_{\epsilon})\\
 =&2 \Re \left(\int \frac{\delta \E}{\delta \overline \phi} \partial_j \overline \phi dx_1+ \int\frac{\delta \E}{\delta \overline \sh} (
 \partial_j \overline \sh)dx_1 dx_2 \right)
 \end{align*}
 together with
\eqref{varphi}, \eqref{vark}
 leads to the conservation
\begin{align*}
\frac{d}{dt} \left( N \int \Im \left( \phi \overline{\partial_j \phi} \right) dx_1 + \int \Im \left( \sh \overline {\partial_j\sh}\right) dx_1 dx_2\right)=0
\end{align*}
thus we define the momentum density
\begin{align*}
 p_{j}(t, x_1)=& \Im \left( \phi \overline{\partial_j \phi} \right)+ \frac{1}{N} \int \Im \left( \sh \overline {\partial_j\sh } \right) dx_2\\
 :=& p_{c, j}(t, x_1) + \frac{1}{N}p_{p, j}(t, x_1)
 \end{align*}
 Finally, using \eqref{varphi}, \eqref{vark} we see that
 \begin{align*}
 \frac{\partial}{\partial t} \E(t) =0
 \end{align*}
 so we define the energy density
 \begin{align*}
&e(t, x_1)= N
\big\vert\nabla\phi_1\big\vert^{2}
\\
&+\frac{N}{2}\int dx_2 v^{N}_{1-2}\vert\phi_1 \phi_2 + \frac{1}{N} (\shh \circ \chh)_{1, 2}\vert^{2}\\
&+\frac{1}{2}\int  dx_2 dx_3v^{N}_{1-2}\vert \phi_1 \shh_{2, 3} + \phi_2 \shh_{1, 3} \vert^2\\
+&\frac{1}{2} \Bigg(\int  dx_2\left\{ \big \vert \nabla_{1, 2} \shh_{1, 2}\big \vert^2 \right \}\\
+&\frac{1}{2N} \int   dx_2
 v^N_{1-2}\Big\{
|(\shh \circ \overline \shh)_{1, 2}|^2 + (\shh \circ \overline \shh)_{1, 1}(\overline \shh \circ  \shh)_{2, 2}
\Big\}
\Bigg) \ .
\end{align*}

\section{A conjecture}
We conjecture  that, if $\phi$, $k$ satisfy the equations of Theorem \eqref{mainthm} and
$\big\vert\psi_{exact}\big>$, $\big\vert\psi_{appr}\big>$, are defined by \eqref{exact}, \eqref{approx}, then,
in the critical case $\beta=1$,
\begin{align*}
\|\big\vert\psi_{exact}\big>-\big\vert\psi_{appr}\big>\|_{\F} \to 0
\end{align*}
as $N \to \infty$, at an explicit rate.

\end{document}